\documentclass{article}
\usepackage{siamltex}

\title{Building Population-Informed Priors for Bayesian Inference Using Data-Consistent Stochastic Inversion}

\author{Rebekah D.~White~\thanks{Sandia National Laboratories, Scientific Machine Learning Department, rebwhit@sandia.gov},
\and John D.~Jakeman~\thanks{Sandia National Laboratories, Optimization and UQ Department, jdjakem@sandia.gov},
\and Tim Wildey~\thanks{Sandia National Laboratories, Computational Mathematics Department,
tmwilde@sandia.gov},
\and Troy Butler~\thanks{University of Colorado Denver, Department of Mathematical and Statistical Sciences, troy.butler@ucdenver.edu}}

\begin{document}
\maketitle

\newcommand{\slugmaster}{%
\slugger{juq}{xxxx}{xx}{x}{x--x}}

\begin{abstract}
Bayesian inference provides a powerful tool for leveraging observational data to inform model predictions and uncertainties.
However, when such data is limited, Bayesian inference may not adequately constrain uncertainty without the use of highly informative priors. 
Common approaches for constructing informative priors 
typically rely on either assumptions or knowledge of the underlying physics, which is not always available.
In this work, we consider the scenario where data are available on a population of assets/individuals, which occurs in many problem domains such as biomedical or digital twin applications, and leverage this population-level data to systematically constrain the Bayesian prior and subsequently improve individualized inferences.
The approach proposed in this paper is based upon a recently developed technique known as data-consistent inversion (DCI) for constructing a pullback probability measure.
Succinctly, we utilize DCI to build population-informed priors for subsequent Bayesian inference on individuals. 
While the approach is general and applies to nonlinear maps and arbitrary priors, we prove that for linear inverse problems with Gaussian priors, the population-informed prior produces an increase in the information gain as measured by the determinant and trace of the inverse posterior covariance. 
We also demonstrate that the Kullback-Leibler divergence often improves with high probability.
Numerical results, including linear-Gaussian examples and one inspired by digital twins for additively manufactured assets, indicate that there is significant value in using these population-informed priors.
\end{abstract}

\section{Introduction}\label{sec:Intro_new}

A key objective of modern computational science and engineering is to use observational data to inform computational models in pursuit of data-informed,
physics-based predictions with quantified uncertainties.
Bayesian inference provides a powerful tool for leveraging observational data and prior knowledge to inform model uncertainties.
The solution to the Bayesian inverse problem is known as the \textit{posterior (probability) distribution}, which is a conditional density quantifying epistemic (reducible) uncertainty in the model parameters given noisy observational data~\cite{Jaynes_IEEETSSC_1968,stuart_2010, Tarantola}.
However, a key challenge in applying Bayesian inference to practical problems is appropriately specifying a prior probability measure/density that quantifies one's prior knowledge without being overly restrictive~\cite{carlin_bayes_2000}.
Often, an approach based on the `principle of insufficient reason,' can be used to specify an uninformative prior (e.g., a uniform density), so that the posterior distribution is primarily influenced by the data~\cite{gelman_bayesian_2013, price_uninformative_2002,van_dongen_prior_2006}.
However, such an approach is inadequate when insufficient data exist to substantially inform the posterior density, motivating the use of the highly informative priors.

%

In this work, we consider problems where the prior is intended to represent the inherent variability (a type of irreducible, i.e., aleatoric, uncertainty) in a population of possible parameter values from which a parameter of interest has been drawn~\cite{gelman_bayesian_2013}. 
The objective is to construct a population-informed prior to facilitate Bayesian inference for the parameters associated with individual-level data.
For example, in biomedical applications where invasive tests on individuals are required, leveraging previous test data on a related population of patients can help constrain or inform Bayesian inference for an individual patient.
Similarly, data from destructive tests on a population of additively manufactured components can be utilized when performing inference on an individual component. 
In such scenarios, the population-level information can provide more objective and informative prior knowledge regarding the individual of interest compared to priors formed using a subjective state of knowledge.

Using population data to construct a prior requires the ability to characterize uncertainty on data obtained from the population. 
Hierarchical Bayesian inference~\cite{gelman_bayesian_2013} and empirical Bayes~\cite{carlin_bayes_2000, berger_statistical_1985, petrone_empirical_2014} are two potential approaches in which prior probability distributions are estimated from the observational data.
In other words, they represent alternative methods for incorporating population-level information into Bayesian inference. 
Empirical Bayes aims to leverage the observational data to inform both the prior and inference parameters.
This is often done through a parametric specification of the prior whose hyperparameters are inferred through marginal maximum likelihood estimation~\cite{berger_statistical_1985, petrone_empirical_2014}.
Given sufficient data, nonparametric approaches can be employed that estimate the empirical cumulative density function (cdf) of the prior from the empirical cdf of the data~\cite{Robbins_1956}.
We note that marginal likelihood estimation, especially for a likelihood depending on computationally intensive models (such as differential equations), poses a non-trivial optimization problem.
Hierarchical Bayesian inference represents a similar, but fully Bayesian approach, wherein hyperpriors on the hyperparameters are specified and informed from data.
This provides a more rigorous but often more computationally complex approach to data-driven prior specification~\cite{petrone_empirical_2014}.  
Indeed, the computational complexity of hierarchical and empirical Bayesian approaches increases with the dimension of the space of hyperparameters, making data-driven inference of priors challenging for large-scale problems.
In addition, such parametric approaches restrict the resulting posteriors to the parametric family defining the prior, which may be unsuitable for certain applications.
%


%

In this work, we propose using data-consistent inversion (DCI)~\cite{wildey_2018} to estimate population-informed priors that constrain Bayesian inverse problems.
DCI is a nonparametric framework that can be used to estimate an \textit{updated (probability) distribution}, which represents the inherent population variability given observational data on the population and an initial distribution on the population-level parameters.
The updated density, derived and analyzed in \cite{wildey_2018}, is a pullback of the observed distribution, i.e., the push-forward of the updated density through the computational model matches the distribution of the observed population data.
While there exist alternative methods to solve such stochastic inverse problems, such as those based on measure-theoretic principles~\cite{breidt_measure-theoretic_2011} as well as more recent work utilizing gradient flows to approximate solutions~\cite{li2024stochasticinverseproblemstability}, we leverage DCI as it has several attractive features.
Namely, DCI guarantees existence and uniqueness (up to the choice of initial density) as well as stability with respect to perturbations in the various distributions utilized in its construction.
Moreover, it has been shown that updated densities converge when using sequences of converging approximate models~\cite{butler2018convergence,butler2022p} or density approximations~\cite{butler2025stability}, and the approach extends easily to non-deterministic models~\cite{butler2020data}, matching empirical cdfs~\cite{bergstrom2024distributions}, or incorporating multiple models of differing fidelity~\cite{bruder2020data}.

%

\subsection{Contributions}
This work presents a novel approach that utilizes DCI to estimate population-informed priors from population data to better constrain Bayesian inference on an individual.
This combination of inverse problems is motivated by the observation that standard Bayesian inference updates the posterior only in directions informed by the individual data, 
while leveraging DCI allows one to update directions informed by the population-level data.
The benefit in combining these approaches is clear when these directions are complementary (or even orthogonal), but we demonstrate that this approach is also beneficial in reducing uncertainty when these directions are aligned, which can occur, for instance, if both the population and individual data come from the same type of measurement.
Moreover, the framework we present applies to linear and nonlinear models as well as to non-Gaussian prior/posterior probability measures.

While the framework applies to general models and measures, we provide a thorough theoretical study of the linear-Gaussian case, which produces closed-form analytical expressions for the inverse problem solutions and information gain.  
These analytical formulations provide tremendous insight on the impact of using population data to better inform uncertain model parameters.
We then move beyond the linear-Gaussian case and demonstrate the effectiveness of this approach using a computational model for an additively-manufactured ``dog-bone'' structures.
Through these examples, we show that using DCI to estimate population-informed priors for Bayesian inference improves the quantification of uncertainty in the individual parameters and often increases the information gained relative to Bayesian inference with standard, uninformed prior specification.
Overall, this work provides a unique approach for incorporating observational data from a population into individualized assessments.  
The main contributions of this work are summarized as follows: 
\begin{itemize}
    \item A novel combination of two types of inverse techniques, DCI and Bayesian inference, leveraging population- and individual-level data to enhance inferences on individuals.
    \item Theoretical results for the linear-Gaussian case that show the model parameter uncertainty is reduced when using the population-informed prior in comparison to standard Bayesian inference.
    \item Numerical demonstrations using analytical and computational mechanics models (inspired by digital twin applications) to demonstrate the effectiveness of this approach.
\end{itemize}

\subsection{Outline}

The remainder of the paper is organized as follows:
\Cref{sec:notation_prob_form} summarizes the  notation and terminology necessary to describe the Bayesian and data-consistent inverse approaches.
\Cref{sec:pi_priors} presents the conceptual and algorithmic framework for constructing the population-informed prior/posterior through the combination of inverse techniques.
Theoretical analysis of the linear-Gaussian case is also included in~\Cref{sec:pi_priors}.
Computational results are presented in~\Cref{sec:computational_results} for both the linear-Gaussian case as well as a nonlinear model.
Conclusions are drawn in \Cref{sec:conclusions}.

\section{Summary of Bayesian and Data-Consistent Inversion}
\label{sec:notation_prob_form}

To properly define the forward and inverse problems considered in this paper, we begin with some notation used throughout this work.
First, let $(\Lambda, \mathcal{B}_{\Lambda}, \mu_{\Lambda})$ denote the measure space associated with model parameters while $(\mathcal{D}_p, \mathcal{B}_{\mathcal{D}_p}, \mu_{\mathcal{D}_p})$ and $(\mathcal{D}_i, \mathcal{B}_{\mathcal{D}_i}, \mu_{\mathcal{D}_i})$ denote measure spaces of data at the population- and individual-levels associated with model observables, respectively. 
In this work, we assume $\Lambda \subset \mathbb{R}^{\pdim}$, $\mathcal{D}_p\subset\mathbb{R}^{\ddim_p}$, and $\mathcal{D}_i\subset\mathbb{R}^{\ddim_i}$ for positive integers $\pdim$, $\ddim_p$, and $\ddim_i$.
We further assume that $\mathcal{B}_{\Lambda}$, $\mathcal{B}_{\mathcal{D}_p}$, and $\mathcal{B}_{\mathcal{D}_i}$ are the Borel $\sigma$-algebras inherited from the respective metric topologies and that $\mu_{\Lambda}$, $\mu_{\mathcal{D}_p}$ and $\mu_{\mathcal{D}_i}$ are the associated volume (typically Lebesgue) measures.

The mappings between the various spaces play a critical role. 
Let $\qmap:\pspace\to\mathcal{D}_p$ denote the map that takes parameter values into the space of observables that we can measure over the population, which is assumed to be measurable and piecewise smooth.
Similarly, let $\imap:\pspace \rightarrow \dspace_i$ denote the map that takes parameter values into the space of observables that we can measure for each individual of the population, which is also assumed to be measurable and piecewise smooth. 
In this work, we explore the following distinct cases for population informed inference: (i) the population and individual maps are identical, i.e., $\qmap(\param) = \imap(\param)$, and (ii) the more typical case where these maps differ.
Note that differing maps could, for example, represent scenarios in which the same underlying experiment is performed, but the population-level observed QoI differ from the individual QoI. 

\subsection{Bayesian Inverse Problem Formulation}
\label{sec:bayes_form}
In this section, we describe the standard Bayesian inference approach that leverages data on an individual to estimate epistemic uncertainty. 
The standard Bayesian formulation often begins by specifying a prior probability measure $\priormeas$ on $(\Lambda, \mathcal{B}_{\Lambda})$.
When this measure is absolutely continuous with respect to $\mu_\pspace$, it admits a corresponding probability density function (PDF) $\priordens$. 
The prior is an initial quantification of uncertainty in one's knowledge or belief of any particular fixed value of the parameter vector being the true value before data are collected on the individual.
Subsequently, Bayes' rule is leveraged to weigh these prior beliefs along with data collected on the individual, denoted by $\data\in\mathcal{D}_i$, to estimate the posterior distribution on the model parameters given as  
\begin{equation}\label{eq:stand_bayes}
\postdens(\param) = \priordens(\param) \frac{\likedens(\data | \param)}{C},
\end{equation}
where $\likedens(\data | \param)$ is a given (data-)likelihood function and
\[C  =\int_\pspace \likedens(\data | \param) \ d\priormeas\]
is a normalizing constant often referred to as the evidence.
The posterior distribution can be interpreted as providing the relative likelihood that a given parameter value could have produced the observed data on the individual.

In many applications of interest, data on the individual may be limited due to the potential high-cost or inability to measure many individual data values (thus limiting the size of $m_i$).
Furthermore, even when a large amount of individual data are available (i.e., $m_i$ is large), an ill-conditioned data map $f_i$ (often determined by analyzing the singular values of its Jacobian) can lead to high-correlation in the data, which reduces the effective dimension of the data space and necessitates the need for informative priors to constrain the Bayesian inverse problem. 
Even in situations where the model and reality coincide,  \cite{stuart_2010} notes that for the under-determined case, ``the posterior measure converges to a Gaussian measure whose support lies on a hyperplane embedded in the space where the unknown [parameter] lies.'' 
Additionally, \cite{Bochkina2012TheBM} 
states that ill-posed problems may violate the assumptions of the the famous Bernstein-von Mises theorem, which is often used in frequentist arguments to justify the approximation of the posterior by a Gaussian distribution and prove convergence to a Dirac delta in the limit of infinite data.
Such convergence results are often interpreted as implying the epistemic uncertainty vanishes (although this argument generally requires the reality and model to agree near the unknown value).
At a high-level, the remarks of \cite{Bochkina2012TheBM,stuart_2010} indicate that for practical problems, data may not inform all directions of the uncertain parameter space, even in the limit of infinite data and with perfect agreement between the model and reality.
Such an issue is further exacerbated by the relatively slow ${\cal O}(1/\sqrt{m_i})$ rate of convergence when standard Monte Carlo schemes are utilized.
For these reasons, informative priors have the potential to greatly improve Bayesian inference. 

\subsection{Data-Consistent Inversion}
\label{sec:DCI_form}

The DCI formulation seeks to constrain the aleatoric (irreducible) uncertainty of data on a population using the pullback of an observed probability measure $\obsmeas$ defined on $(\dspace_p, \mathcal{B}_{\dspace_p})$ using $f_p$. 
In other words, given an observed probability measure, $\obsmeas$, on $(\mathcal{D}_p, \mathcal{B}_{\mathcal{D}_p})$, DCI seeks a probability measure $\mathbb{P}_\pspace$ on $(\Lambda, \mathcal{B}_{\Lambda})$ such that the push-forward of $\mathbb{P}_\pspace$ through $\qmap$ matches the observed measure, i.e.,
\begin{equation}
\obsmeas(A) = \mathbb{P}_\pspace(\qmap^{-1}(A)) \quad \forall A \in \mathcal{B}_{\mathcal{D}_p}.\label{eq:consistent}
\end{equation}

Since pullback measures are, in general, not unique, we follow~\cite{wildey_2018} to construct a DCI solution by first defining
an {\em initial} probability measure $\initmeas$ on $(\Lambda, \mathcal{B}_{\Lambda})$, which we seek to update in such a way that a pullback is uniquely defined with respect to this initial measure.
Note that the role of the initial probability measure is fundamentally different than the prior measure in a Bayesian setting since it represents an initial quantification of aleatoric uncertainty in the model parameters across the population.
Given an initial probability measure, a forward uncertainty quantification problem is solved to construct a predicted probability measure on the population-level data space $(\dspace_p, \mathcal{B}_{\dspace_p})$ given by
\begin{equation*}
    \predmeas(A) := \initmeas(f_p^{-1}(A)),  \quad \forall A \in \mathcal{B}_{\dspace_p}.
\end{equation*}
Clearly, if $\predmeas$ matches $\obsmeas$, then $\initmeas$ defines a DCI solution; we do not expect this to happen in practice.
Instead, an update to $\initmeas$ is possible under reasonable assumptions. 
We first assume that the initial, predicted, and observed measures are all absolutely continuous with respect to their respective volume measures so that they admit corresponding PDFs, $\initdens$, $\preddens$, and $\obsdens$, respectively.
To guarantee the {\em existence} of a solution, we follow~\cite{wildey_2018}, 
and make the following {\em predictability assumption} regarding these densities:
\begin{assumption}[Predictability Assumption]\label{assump:dom}
\normalfont
There exists a constant $C_p>0$ such that
\[\obsdens(q)\leq C_p\preddens(q), \quad \text{for a.e.} \hspace{0.1cm} q\in\dspace_p.\]
\end{assumption}

Assumption~\ref{assump:dom} implies that the observed probability measure is absolutely
continuous with respect to the push-forward of the initial probability measure.
This assumption is also related to the ability to construct numerical approximations of the updated density defining the DCI solution since the constant $C_p$ is the same constant utilized/estimated when performing rejection sampling.
Under the predictability assumption, a disintegration theorem is used to give a unique update to the initial probability measure that defines a DCI solution as
\begin{equation}\label{eq:upmeas}
\upmeas(A) := \int_{\mathcal{D}}\Big(\int_{A\cap \qmap^{-1}(d)} \initdens(\param)\frac{\obsdens(\qmap(\param))}{\preddens(\qmap(\param))}d \mu_{\Lambda,d}(\param)\Big)d\mu_{\mathcal{D}_p}(d).
\end{equation}
Adopting the convention established in~\cite{bruder2020data,butler2018convergence}, we refer to this probability measure as the {\em updated} measure on $(\Lambda, \mathcal{B}_{\Lambda})$, and its associated updated density is given by
\begin{eqnarray}\label{eq:upd_dens}
\updens(\param) = \initdens(\param)r(\param), \quad \text{where} \quad r(\param) =\frac{\obsdens(\qmap(\param))}{\preddens(\qmap(\param))}.
\end{eqnarray}

In practice, the ratio $r(\param)$ updates the initial density to construct a solution to the stochastic inverse problem.  
The sample average of the ratio (computed from samples generated by the initial distribution) also provides a useful diagnostic for numerical validation of the predictability assumption since the expected value corresponds to integrating the updated density and should therefore be unity.
In other words, if the predictability assumption holds, then the updated density is in fact a density implying that
\begin{equation}\label{eq:E_init_r}
    1 = \int_\pspace \updens(\lambda)\, d\mu_\pspace = \int_\pspace r(\lambda)\initdens(\lambda)\, d\mu_\pspace = \mathbb{E}_\text{init}(r(\lambda)).
\end{equation}

While, in general, $\initdens$ and $\priordens$ may be different, we take them to be the same throughout this paper for the sake of simplicity and making fair comparisons.
However, it is worth noting that in~\cite{wildey_2018} the initial and updated measures are referred to as prior and posterior measures to reflect the fact that~\eqref{eq:upmeas} is derived by combining the disintegration theorem (e.g., see~\cite{Bogachev_vol1,chang_1997}) and Bayes' theorem.
As with the works that chronologically follow \cite{wildey_2018} (e.g., see \cite{bruder2020data,butler2020optimal,butler2020data,butler2022p}), this work uses {\em initial/updated} instead of {\em prior/posterior} to emphasize the fact that the inverse problem defined by~\eqref{eq:consistent} is fundamentally different from the classical Bayesian inverse problem.
In particular, in the limit of infinite population data, the characterization of the observed density improves, and therefore the DCI solution also improves, but does not converge to a Dirac delta.
We refer the interested reader to~\cite{wildey_2018} for a formal derivation of~\eqref{eq:upmeas}, the theory of existence, uniqueness, and stability of $\upmeas$, and also a direct comparison of $\updens$ with the solution to a classical Bayesian inverse problem.

\subsection{An Illustrative Example of Data-Consistent Inversion}
\label{sec:vis_DCI}
To visualize how DCI leverages information regarding a population to construct a pullback probability measure, consider the following mathematical model of a population 
\begin{eqnarray}\label{eq:pop_mod_0_1}
    f_p({\bm \lambda}) = \begin{bmatrix}
        0 & 1 
    \end{bmatrix} \begin{bmatrix}
        \lambda_1 \\ \lambda_2
    \end{bmatrix},
\end{eqnarray}
where ${\bm \lambda} = [\lambda_1, \lambda_2]$ are the uncertain parameters we wish to infer from population-level data.
Additionally, assume that the \textit{true} population-generating probability density on these uncertain parameters is 
\begin{eqnarray}\label{eq:pop_gen_0_1}
    \popgendens(\param) \sim \mathcal{N}\left(\begin{bmatrix}
    0.2 \\ 0.3
\end{bmatrix}, \begin{bmatrix}
    0.06 & 0 \\ 0 & 0.06
\end{bmatrix}
\right).
\end{eqnarray}
This true density, shown in the left plot of~\Cref{fig:tri_plot}, generates the observed density on the population data given as $\obsdens \sim \mathcal{N}\left(0.3, 0.06\right)$.
In practice, one would not have access to this true population-generating distribution, but utilizing it in this illustrative example allows us to demonstrate how population-level data informs certain directions of this true distribution due to the ill-posedness of the inverse problem.

Using the initial density
\begin{eqnarray}\label{eq:init_0_1}
    \initdens(\param) \sim \mathcal{N}\left(\begin{bmatrix}
    0.4 \\ 0.0
\end{bmatrix}, \begin{bmatrix}
    0.15 & 0 \\ 0 & 0.15
\end{bmatrix}
\right),
\end{eqnarray}
which is depicted in the middle plot of~\Cref{fig:tri_plot},
 and the population model given in~\eqref{eq:pop_mod_0_1}, yields the updated density 
\begin{eqnarray}\label{eq:up_dens_0_1}
    \updens(\param) \sim \mathcal{N}\left(\begin{bmatrix}
    0.4 \\ 0.3
\end{bmatrix}, \begin{bmatrix}
    0.15 & 0 \\ 0 & 0.06
\end{bmatrix}
\right)
\end{eqnarray}
shown in the right plot of~\Cref{fig:tri_plot}.
\begin{figure}[h]
    \centering
    \includegraphics[trim={0 0 0 1cm},clip,width=0.96\textwidth]{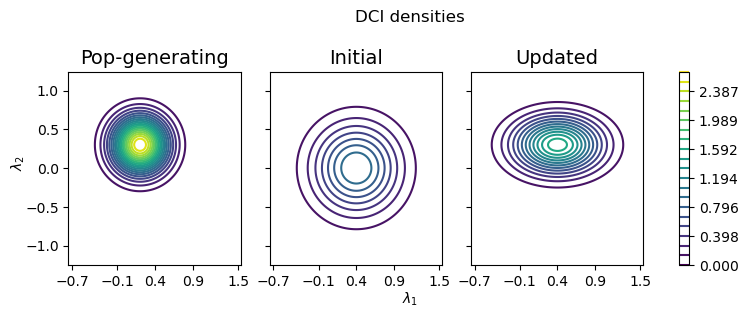}
    \caption{The \textit{true} population-generating density $\popgendens$ (left), the initial density $\initdens(\param)$ (middle), and the updated density $\updens(\param)$ (right).}
    \label{fig:tri_plot}
\end{figure}

The push-forward of this updated density through the population model defined by~\eqref{eq:pop_mod_0_1} is compared to the push-forward of the initial density in \Cref{fig:pf_comp}.
\begin{figure}[h]
    \centering
    \includegraphics[width=0.48\textwidth]{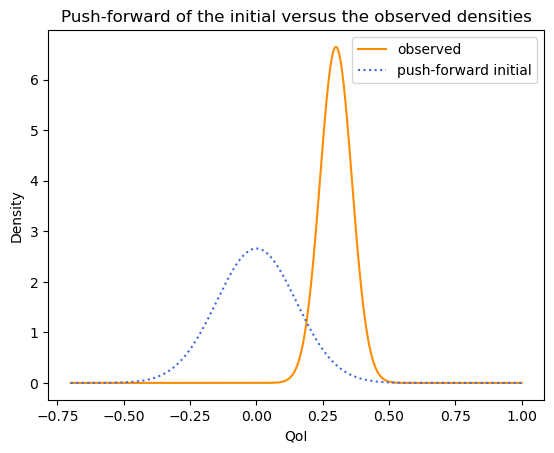}
    \includegraphics[width=0.48\textwidth]{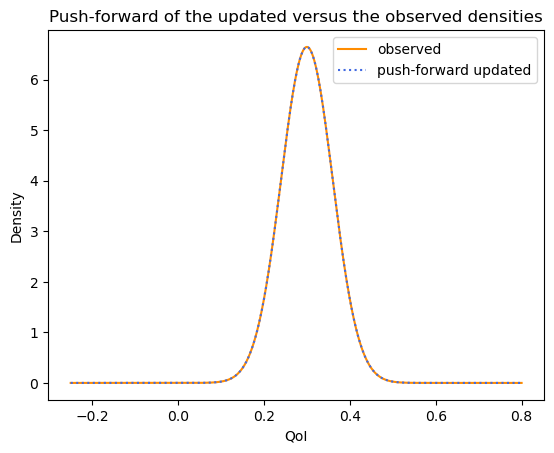}
    \caption{Comparison of the push-forwards (through the QoI map $f_p({\bm \lambda})$) of the initial density $\initdens$ (left) versus the updated density (right).}
    \label{fig:pf_comp}
\end{figure}
From this figure, it is clear that the updated density $\updens(\param)$ is precisely the density that \textit{updates} the initial such that this updated density is \textit{consistent} with the observed density. %
It is worth emphasizing that while the updated density does solve the inverse problem in the sense that it is a pullback probability density, the updated density will often not be identical to the true population-generating distribution (unless the map $f_p(\param)$ is one-to-one).
For instance, the population model given in~\eqref{eq:pop_mod_0_1} only informs the $\lambda_2$ direction, and consequently, the initial density is updated only in the $\lambda_2$ direction as seen in~\eqref{eq:up_dens_0_1} and~\Cref{fig:tri_plot}.

\subsection{Practical Estimation and Sampling of Solutions}

Before providing the algorithmic framework for integrating Bayesian and DCI methods, we end this current section with a high-level summary of typical options available for the practical estimation and sampling of individual posterior and updated distributions.
Methods such as Markov Chain Monte Carlo (MCMC) are often utilized to sample from the posterior distribution in Bayesian frameworks to avoid estimation of the normalizing constant (i.e., the evidence term) as this typically requires many model evaluations.
This is especially useful when the support of the likelihood is a relatively small subset of the support of the prior.
However, for the sake of simplicity,
we generate samples from the prior and use standard Monte Carlo estimation to approximate the normalizing constant and other values necessary for computing quantities such as the Kullback-Leibler (KL) divergence as shown in ~\Cref{sec:dci_up}.
Nevertheless, the algorithmic framework and examples presented in this work are amenable to approaches that do not estimate such constants. 

The approximation and sampling of the updated density within the DCI framework requires the estimation of the predicted density. In contrast to the denominator of the Bayesian posterior density, which is constant, the predicted density appearing in the denominator of the updated density (as seen in~\eqref{eq:upd_dens}) is a function of the model parameters. In this paper we use kernel density estimation (KDE) to estimate the predicted density; this is tractable as the dimensionality of the KDE is determined by the number of quantities of interest and not the number of model parameters.
Once the predicted density is approximated, it is possible to sample from an updated density via MCMC methods. However, as mentioned above, to keep the methods simple and straightforward to implement and reproduce, we follow the Monte Carlo approach used in~\cite{wildey_2018,butler2020optimal,butler2020data, tran2021solving} based on rejection sampling to draw independent identically distributed samples from the updated density.

\section{Using DCI to Build Population-Informed Priors}
\label{sec:pi_priors}

Throughout this work, we refer to solving the Bayesian inverse problem using a population-informed prior as \textit{population-informed inference} and the corresponding solution as the \textit{population-informed posterior}.
In contrast, the Bayesian posterior utilizing standard prior specification, as in~\eqref{eq:stand_bayes}, will be referred to as the \textit{standard posterior} computed according to \textit{standard Bayesian inference}.

To construct a population-informed posterior, we leverage the fact that DCI solves a fundamentally different problem than Bayesian inference and utilize the updated density as a population-informed prior. 
In other words, we first use DCI to compute $\updens(\param)$ as the population-informed prior, which subsequently defines the corresponding population-informed posterior as 
\begin{eqnarray}\label{eq:pop_post}
    \postupdens(\param) := \updens(\param)\frac{\likedens(\data | \param)}{\tilde{C}} = \initdens(\param)\frac{\obsdens(\qmap(\param))}{\preddens(\qmap(\param))}\frac{\likedens(\data | \param)}{\tilde{C}},
\end{eqnarray}
where, by direct substitution and re-arrangement of terms, the population-informed evidence denoted by $\tilde{C}$ is written as 
\begin{eqnarray}\label{eq:tilde_c}
    \tilde{C}  := \int_\pspace \frac{\obsdens(\qmap(\param))}{\preddens(\qmap(\param))} \likedens(\data | \param) \initdens(\param) \, d\mu_\pspace = \mathbb{E}_\text{init}\left[r(\param) \likedens(\data | \param)\right].
\end{eqnarray}
From~\eqref{eq:pop_post}, it is evident that once the predicted density is constructed for a given observed density (on the population-level data) and likelihood (on the individual-level data), we have all the ingredients necessary to sample from the population-informed posterior given samples from the initial density. 
\Cref{alg:alg_pop_post} summarizes the rejection sampling procedure utilized in this work to generate independent identically distributed samples from the population-informed posterior. 
It is worth noting that the initial samples utilized for approximating both the predicted density and the population-informed evidence term are reused in the rejection sampling step.
\begin{algorithm}[h]
    \DontPrintSemicolon
    \SetAlgoLined
    \KwInput{$\{\param^{(j)} \}_{j=1}^{N} \sim \initdens$, individual data ${\bm y}$, a population-level observed density $\obsdens$, and an individual-level likelihood $\likedens$.}
    \textbf{Pre-processing steps:} \\
    \begin{enumerate}
        \item Use a KDE on $\{f_p^{(j)} := f_p(\param^{(j)})\}_{j=1}^N$ to estimate $\preddens$. 
        \item Check if the predictability assumption holds using the diagnostic of~\eqref{eq:E_init_r} to verify
        \begin{eqnarray*}
            1 = \mathbb{E}_\text{init}(r(\param)) \approx \frac{1}{N}\sum_{j=1}^N \frac{\obsdens(f_p^{(j)})}{\preddens(f_p^{(j)})}. 
        \end{eqnarray*}
        \item Approximate the population-informed evidence of~\eqref{eq:tilde_c} as
        \begin{eqnarray*}
            \tilde{C} = \mathbb{E}_\text{init}\left[r(\param) \likedens(\data | \param)\right] \approx \frac{1}{N} \sum_{j = 1}^{N} \frac{\obsdens(f_p^{(j)})}{\preddens(f_p^{(j)})} \likedens(\data | \param^{(j)}).
        \end{eqnarray*}
        \item Let $\alpha(\param):=\frac{r(\param)\likedens(\data | \param^{(j)})}{\tilde{C}}$ and compute $\{ \alpha^{(j)} :=\alpha(\param^{(j)}) \}_{j=1}^{N}$.
        \item Estimate $M = \max_{\param} \alpha(\param) \approx \max_{1\leq j\leq n} \alpha^{(j)}$
    \end{enumerate}
    \textbf{Rejection sampling:} \\
    Set $S = 0$\;
    \For{ $j = 1, \dots, N$}
    {Generate a random number $t^{(i)} \sim \mathcal{U}([0,1])$ and compute $\eta^{(j)} = \alpha^{(j)}/M$. \;
    \eIf{$t^{(j)} < \eta^{(j)}$} 
    {Accept and set $S=S+1$ and $\param_{\text{accept}}^{(j)}=\param^{(j)}$.}
    {Reject $\param^{(j)}$.}
    }
    \KwOutput{Samples from the population-informed posterior $\{ \param_{\text{accept}}^{(j)} \}_{j =1}^{S} \sim \postupdens$.} 
    \caption{Drawing samples from the population-informed posterior $\postupdens$}
     \label{alg:alg_pop_post}
\end{algorithm}

The remainder of this section is mostly dedicated to the analysis of the linear-Gaussian case.
We first compare the population-informed posterior~\eqref{eq:pop_post} to the standard Bayesian posterior~\eqref{eq:stand_bayes} for the case where both the population and individual parameter-to-observable maps are linear and the likelihood, prior, and initial densities are Gaussian.
Such assumptions result in closed-form expressions for the updated and posterior densities, which allow us to prove that by metrics such as the trace and determinant of the inverse covariance, population-informed inference is more informative than standard Bayesian inference.
Finally, we compare the information gain, as quantified by the KL divergence, that arises from the update of the prior to the standard and population-informed posteriors.  
We demonstrate that in general (nonlinear and non-Gaussian), the expressions for these KL divergences are quite similar.
However, even in the linear-Gaussian case, we explain why population-informed inference may not result in greater information gained as measured by the KL divergence than the standard Bayesian inference.

\subsection{The Linear-Gaussian Updated Density}
\label{sec:dci_up} 
In this subsection, we state the basic assumptions for the use and analysis of DCI for linear-Gaussian models.
Assume that the parameter-to-population-observable map is given as $f_p(\bml) = {\bm A}\bml$ for some full-rank ${\bm A} \in \mathbb{R}^{\ddim_p\times\pdim}$, the population-level observed density is given as $\obsdens \sim \mathcal{N}(\bar{\bm f_p}, {\bm \Gamma}_{\text{obs}})$ (where $\bar{\bm f_p}$ denotes the observed mean on the QoI), and the initial density is given as $\initdens(\bml) \sim \mathcal{N}({\bm \mu}_{\text{in}}, {\bm \Gamma}_{\text{in}})$.
Note the use of subscripts on the hyperparameters defining the means and covariances are done for both clarity and notational simplicity since we often refer to the inverses of the covariance matrices. 
Also note that the covariances are of different shapes given the different dimensions of the spaces, i.e., ${\bm \Gamma}_{\text{obs}}\in\mathbb{R}^{\ddim_p\times\ddim_p}$ and ${\bm \Gamma}_{\text{in}}\in\mathbb{R}^{\pdim\times\pdim}$.
In Section~\ref{sec:proof_linGaus}, we discuss sufficient conditions on both the covariances and the map $f_p$ that guarantee the predictability assumption holds. 
For now, assuming that the predictability assumption holds, the resulting updated density is Gaussian (e.g., see see~\cite{marvin2018,PILOSOV2023115906} for derivations) and written as $\updens(\bml) \sim \mathcal{N}({\bm \mu}_{\text{up}}, {\bm \Gamma}_{\text{up}})$ (with ${\bm \Gamma}_{\text{up}}\in\mathbb{R}^{\pdim\times\pdim}$) where  
\begin{eqnarray}\label{eq:cov_up}
    {\bm \Gamma}_{\text{up}} &=&  \left( {\bm A}^\top {\bm \Gamma}_{\text{obs}}^{-1}{\bm A} + {\bm \Gamma}_{\text{in}}^{-1} - {\bm A}^\top\left({\bm A}{\bm \Gamma}_{\text{in}}{\bm A}^\top \right)^{-1}{\bm A}
    \right)^{-1},\\\label{eq:mu_up}
    {\bm \mu}_{\text{up}} &=& {\bm \mu}_{\text{in}} + {\bm \Gamma}_{\text{up}} {\bm A}^\top {\bm \Gamma}_{\text{obs}}^{-1}\left(\bar{\bm f_p}-{\bm A}{\bm \mu}_{\text{in}}\right).
\end{eqnarray}

\subsection{The Linear-Gaussian Standard and Population-Informed Posteriors}
\label{sec:bays_post}

A closed-form expression for the standard Bayesian posterior in the linear-Gaussian case is well-known (e.g., see~\cite{stuart_2010}).
We focus this section on deriving the population-informed posterior for the linear-Gaussian case where the parameter-to-individual-observable map is given by $f_i(\bml) = {\bm B}\bml$ for some full-rank ${\bm B}\in\mathbb{R}^{\ddim_i\times \pdim}$.
In cases where the same experiment and measurements are obtained for both the individual- and population-level data, ${\bm B} = {\bm A}$; otherwise, we expect that ${\bm B} \neq {\bm A}$.

To provide closed-form expressions for the individual and population informed posteriors, assume $\priordens(\bml) \sim \mathcal{N}({\bm m}_{\text{pr}}, {\bm \Gamma}_{\text{pr}})$
and that the individual data satisfy 
\begin{eqnarray*}
    {\bm y} = f_i(\bml) + {\bm \eta} = {\bm B}\bml + {\bm \eta}, 
\end{eqnarray*}
where ${\bm \eta} \sim \mathcal{N}\left({\bm 0}, {\bm \Gamma_{\text{noise}}}\right)$ and ${\bm \Gamma_{\text{noise}}} \in \mathbb{R}^{\ddim_i \times \ddim_i}$.
Given this Gaussian noise assumption, the likelihood is given as $\pi_{\text{lik}}\left({\bm y} | \bml\right) \sim \mathcal{N}\left({\bm B}\bml, {\bm \Gamma}_{\text{noise}}\right)$, and the resulting
standard Bayesian posterior is Gaussian, i.e., $\postdens(\bml) \sim \mathcal{N}\left({\bm m}_{\text{post}}, {\bm \Gamma}_{\text{post}}\right)$, where
\begin{eqnarray}\label{eq:linGaus_cov}
    {\bm \Gamma}_{\text{post}} &=& \left({\bm B}^\top {\bm \Gamma}_{\text{noise}}^{-1} {\bm B} + {\bm \Gamma}_{\text{pr}}^{-1}\right)^{-1},\\\label{eq:linGaus_mean}
    {\bm m}_{\text{post}} &=& {\bm m}_{\text{pr}} + {\bm \Gamma}_{\text{post}}{\bm B}^\top {\bm \Gamma}_{\text{noise}}^{-1} \left({\bm y} - {\bm B}{\bm m}_{\text{pr}}\right) .
\end{eqnarray}
From~\Cref{sec:dci_up}, it follows that the population-informed posterior is also Gaussian.
However, in contrast to the standard posterior, we have that $ \postupdens(\bml) \sim \mathcal{N}(\tilde{\bm m}_{\text{post}}, \tilde{\bm \Gamma}_{\text{post}})$, where~\eqref{eq:cov_up} produces the following population-informed posterior covariance
\begin{eqnarray}\label{eq:linGaus_pop_cov}
    \tilde{\bm \Gamma}_{\text{post}} &=& \left({\bm B}^\top {\bm \Gamma}_{\text{noise}}^{-1} {\bm B} + {\bm \Gamma}_{\text{up}}^{-1}\right)^{-1}\\
    \label{eq:linGaus_pop_cov_long}
        &=& \left({\bm B}^\top {\bm \Gamma}_{\text{noise}}^{-1} {\bm B} + \left( {\bm A}^\top {\bm \Gamma}_{\text{obs}}^{-1}{\bm A} + {\bm \Gamma}_{\text{in}}^{-1} - {\bm A}^\top\left({\bm A}{\bm \Gamma}_{\text{in}}{\bm A}^\top \right)^{-1}{\bm A}\right)\right)^{-1},
\end{eqnarray}
and~\eqref{eq:mu_up} produces the following population-informed posterior mean
\begin{eqnarray}
    \label{eq:linGaus_pop_mean}
    \tilde{\bm m}_{\text{post}} &=& {\bm \mu}_{\text{up}} + \tilde{\bm \Gamma}_{\text{post}}{\bm B}^\top {\bm \Gamma}_{\text{noise}}^{-1} \left({\bm y} - {\bm B}{\bm \mu}_{\text{up}}\right) \\
    \label{eq:linGaus_pop_mean_long}
        &=& {\bm \mu}_{\text{up}} + \tilde{\bm \Gamma}_{\text{post}}{\bm B}^\top {\bm \Gamma}_{\text{noise}}^{-1} \left({\bm y} - {\bm B}\left({\bm \mu}_{\text{in}} + {\bm \Gamma}_{\text{up}} {\bm A}^\top {\bm \Gamma}_{\text{obs}}^{-1}\left(\bar{\bm f_p}-{\bm A}{\bm \mu}_{\text{in}}\right)\right)\right).
\end{eqnarray}

\subsection{Impact of Population-Informed Prior on Posterior Covariance}
\label{sec:proof_linGaus}

A variety of metrics involving the posterior distribution are available to us to compare the effectiveness of incorporating population data into the Bayesian inverse problem.
For example, in standard Bayesian inference, the KL divergence is used to quantify the information gain by utilizing observational data to generate the posterior distribution.
Significant deviations of the posterior from the prior are quantified by large KL divergences and interpreted as significant information being gained by solving the Bayesian inverse problem.
Alternatively, scalar-valued functions of the posterior covariance can indicate a reduction in parameter uncertainty or analogously an increase in information gain from solving the inverse problem. 
Common scalar-valued functions include the trace and determinant or the covariance (or its inverse, which is referred to as the precision matrix).
Such criteria are often used in optimal experimental design and correspond to D-optimal (i.e., minimizing the determinant of the covariance)~\cite{Alexanderian_S_SISC_2018,Attia_2018} and A-optimal (i.e., minimizing the trace of the covariance)~\cite{doi:10.1137/140992564,Attia_2018,Haber2012} designs, respectively.

For the linear-Gaussian setting, we leverage the closed-form expressions given in~\Cref{sec:bays_post} to prove that population-informed posteriors always produces an equivalent or increased
information gain as measured in terms of the trace and determinant of the inverse posterior covariance matrix.
To see this, we first show that under reasonable assumptions, the inverse covariance of the updated density is symmetric positive definite (SPD). 
Following a similar approach to that given in~\cite{spantini_2015} and later developed for DCI in~\cite{marvin2018}, we define ${\bm Q} = {\bm \Gamma}_{\text{in}}^{1/2}{\bm A}^\top{\bm \Gamma}_{\text{obs}}^{-1/2} \in\mathbb{R}^{\pdim \times \ddim_p}$ whose reduced singular value decomposition (SVD) is given as 
\begin{eqnarray*}
    {\bm Q} = {\bm U}{\bm \Sigma}{\bm V}^\top,
\end{eqnarray*}
where ${\bm U} \in \mathbb{R}^{\pdim \times \ddim_p}$, ${\bm V} \in \mathbb{R}^{\ddim_p \times \ddim_p}$, and ${\bm \Sigma} \in \mathbb{R}^{\ddim_p \times \ddim_p}$ is a diagonal matrix whose entries are the singular values of ${\bm Q}$.

\begin{prop}\label{prop:spd}
\normalfont
If $\qmap$ is linear, ${\bm \Sigma}_{ii} \geq 1$ for all $1\leq i\leq \ddim_p$ (i.e., the singular values of ${\bm Q}$ are all bounded below by unity), and both $\priordens$ and $\obsdens$ are Gaussian with $\initdens(\bml) \sim \mathcal{N}({\bm \mu}_{\text{in}}, {\bm \Gamma}_{\text{in}})$ and $\obsdens \sim \mathcal{N}(\bar{{\bm f}_p}, {\bm \Gamma}_{\text{obs}})$, then ${\bm \Gamma}_{\text{up}}^{-1}$ is symmetric positive definite (SPD).
\end{prop}
\begin{proof}
Using the updated covariance given in~\eqref{eq:cov_up} along with the SVD of ${\bm Q}$, we have that 
\begin{eqnarray}\nonumber
    {\bm \Gamma}_{\text{up}}^{-1} &=& {\bm A}^\top {\bm \Gamma}_{\text{obs}}^{-1}{\bm A} + {\bm \Gamma}_{\text{in}}^{-1} - {\bm A}^\top\left({\bm A}{\bm \Gamma}_{\text{in}}{\bm A}^\top \right)^{-1}{\bm A} \\\nonumber
    &=&
    {\bm \Gamma}_{\text{in}}^{-1/2}\left(\mathbb{I}_\pdim + {\bm U}\left({\bm \Sigma^2} - \mathbb{I}_{\ddim_p}\right){\bm U}^\top\right){\bm \Gamma}_{\text{in}}^{-1/2}
    \\\label{eq:update_prec}
    &=& {\bm \Gamma}_{\text{in}}^{-1} + {\bm \Gamma}_{\text{in}}^{-1/2} {\bm U}{\bm G}{\bm U}^\top{\bm \Gamma}_{\text{in}}^{-1/2}
\end{eqnarray}
where 
$\mathbb{I}_{x} \in \mathbb{R}^{x \times x}$ represents an identity matrix, and ${\bm G} = \left({\bm \Sigma^2} - \mathbb{I}_{\ddim_p}\right)$ is a diagonal matrix whose entries are given by ${\bm G}_{ii} = {\bm \Sigma}^2_{ii} - 1$.
Note that if ${\bm \Sigma}_{ii} = 1$ for all $i$, then the covariance of the updated density is the covariance of the initial distribution.
If ${\bm \Sigma}_{ii}<1$ for any $1\leq i\leq \ddim_p$, then this indicates a violation of the predictability assumption.
Assuming ${\bm \Sigma}_{ii}\geq 1$, we have that ${\bm G}_{ii}\geq 0$.
Therefore, ${\bm \Gamma}_{\text{in}}^{1/2} {\bm U}{\bm G} {\bm U}^\top {\bm \Gamma}_{\text{in}}^{1/2}$ is positive semi-definite and we can write~\eqref{eq:update_prec} generally as 
\begin{eqnarray}\label{eq:prec_spd}
    {\bm \Gamma}_{\text{up}}^{-1} = 
        {\bm \Gamma}_{\text{in}}^{-1} + {\bm Y}{\bm Y}^\top,
\end{eqnarray}
where ${\bm Y} := {\bm \Gamma}_{\text{in}}^{1/2} {\bm U}{\bm G}^{1/2} \in \mathbb{R}^{\pdim \times \ddim_p}$.
Since ${\bm \Gamma}_{\text{in}}^{-1}$ is SPD and ${\bm Y}{\bm Y}^\top$ is positive semi-definite, it follows that ${\bm \Gamma}_{\text{up}}^{-1}$ is also SPD.
%
%
\end{proof}

Next, we show that~\eqref{eq:prec_spd} allows us to prove that the determinant of the population-informed posterior covariance is always less than or equal to the determinant of the posterior covariance for standard Bayesian inference when the initial and prior distributions have matching covariances, which implies that incorporating population-level data into the Bayesian inverse problem results in greater (or equivalent) information gain than the utilization of individual-level data alone.
\begin{prop}\label{prop:det}
\normalfont
If the assumptions in Proposition~\ref{prop:spd} are satisfied and ${\bm \Gamma}_\text{pr} = {\bm \Gamma_\text{in}}$, then
\begin{equation}\label{eq:proof_det_cov}
     \det({\bm \Gamma}^{-1}_{\text{post}}) \leq \det(\tilde{\bm \Gamma}^{-1}_{\text{post}}).
\end{equation}

\end{prop}
\begin{proof}
Consider that
\begin{subequations}\label{eq:det_proof_1}
    \begin{align}
    \det\left(\tilde{\bm \Gamma}^{-1}_{\text{post}} \right) &= \det \left({\bm B}^\top {\bm \Gamma}_{\text{noise}}^{-1} {\bm B} + {\bm \Gamma}_{\text{up}}^{-1}\right) 
    \\\label{eq:det_proof_2}
    &= \det \left({\bm B}^\top {\bm \Gamma}_{\text{noise}}^{-1} {\bm B} + {\bm \Gamma}_{\text{in}}^{-1} + {\bm Y}{\bm Y}^\top\right)
    \\\label{eq:det_proof_3}
    &=\det \left({\bm \Gamma}^{-1}_{\text{post}} + {\bm Y}{\bm Y}^\top\right)
    \\\label{eq:det_proof_4}
    &\geq \det \left({\bm \Gamma}^{-1}_{\text{post}}\right) + \det \left({\bm Y}{\bm Y}^\top\right)
    \\\label{eq:det_proof_5}
    &\geq \det\left({\bm \Gamma}^{-1}_{\text{post}} \right),
    \end{align}
\end{subequations}
where~\eqref{eq:det_proof_2} follows from expression~\eqref{eq:prec_spd}, \eqref{eq:det_proof_3} follows from definition~\eqref{eq:linGaus_cov}, \eqref{eq:det_proof_4} follows from Minkowski’s determinant theorem and the fact that ${\bm \Gamma}^{-1}_{\text{post}}$ and ${\bm Y}{\bm Y}^\top$ are both SPD; finally,~\eqref{eq:det_proof_5} results from the fact $\det({\bm Y}{\bm Y}^\top) = \det({\bm Y})^2 \geq 0$.
%
The result in~\eqref{eq:proof_det_cov} is then an immediate consequence of the above chain of inequalities. 
\end{proof}

%
As with the determinant, the trace of the inverse population-informed posterior covariance will always be greater than or equal to the trace of the standard posterior covariance, indicating that average parameter precision (inverse variance) improves with the utilization of population-level data.
\begin{prop}\label{prop:trace}
\normalfont
If the assumptions in Proposition~\ref{prop:spd} are satisfied and ${\bm \Gamma}_\text{pr}={\bm \Gamma}_\text{init}$, then
\begin{equation}\label{eq:trace_result}
    \text{trace}\left({\bm \Gamma}^{-1}_{\text{post}} \right) \leq \text{trace}\left(\tilde{\bm \Gamma}^{-1}_{\text{post}} \right).
\end{equation}
\end{prop}
\begin{proof}
Consider first that~\eqref{eq:prec_spd}
provides that 
\begin{eqnarray*}
    \text{trace}\left({\bm \Gamma}_{\text{up}}^{-1}\right) = 
    \text{trace}\left({\bm \Gamma}_{\text{pr}}^{-1}\right) + \text{trace}\left({\bm Y}{\bm Y}^\top\right) \geq \text{trace}\left({\bm \Gamma}_{\text{pr}}^{-1}\right),
\end{eqnarray*}
since $\text{trace}({\bm Y}{\bm Y}^\top) = \sum_{i = 1}^{\pdim}\sum_{j =1}^\pdim ({\bm Y}_{ij})^2 \geq 0$.
Therefore, 
\begin{subequations}\label{eq:trace_proof_1}
    \begin{align}
    \text{trace}\left({\bm \Gamma}^{-1}_{\text{post}} \right) 
    &= 
    \text{trace}\left({\bm B}^\top {\bm \Gamma}_{\text{noise}}^{-1} {\bm B} + {\bm \Gamma}_{\text{pr}}^{-1}\right) 
    \\\label{eq:trace_proof_2}
    &= 
    \text{trace}\left({\bm B}^\top {\bm \Gamma}_{\text{noise}}^{-1} {\bm B}\right) + \text{trace}\left({\bm \Gamma}_{\text{pr}}^{-1}\right)\\\label{eq:trace_proof_3}
    &\leq 
    \text{trace}\left({\bm B}^\top {\bm \Gamma}_{\text{noise}}^{-1} {\bm B}\right) + \text{trace}\left({\bm \Gamma}_{\text{up}}^{-1}\right)
    \\\label{eq:trace_proof_4}
    &= 
    \text{trace}\left({\bm B}^\top {\bm \Gamma}_{\text{noise}}^{-1} {\bm B}+ {\bm \Gamma}_{\text{up}}^{-1}\right)
    \\\label{eq:trace_proof_5}
    &=
    \text{trace}\left(\tilde{\bm \Gamma}^{-1}_{\text{post}} \right),
    \end{align}
\end{subequations}
which proves~\eqref{eq:trace_result}.
\end{proof}

Note that the results given by~\eqref{eq:proof_det_cov} and~\eqref{eq:trace_result} hold regardless of whether the population and individual models differ. 
However, they do require that both the individual and population parameter-to-observable maps are linear, the prior and likelihood are Gaussian, and that the prior covariance used in the standard Bayesian inference is equivalent to the initial covariance used in the DCI problem to construct the population-informed prior.

\subsection{Impact of Population-Informed Prior on Information Gain}\label{subsec:infogain}
It is quite common, especially in the context of optimal experimental design, to utilize the KL divergence of the posterior from the prior
to quantify the information gained.
Given two probability densities, $\pi^A$ and $\pi^B$, on $(\pspace,\pborel)$ with $\pi^A$ absolutely continuous with respect to $\pi^B$, the KL divergence of $\pi^A$ from $\pi^B$ is given by
\begin{equation}\label{eq:KL}
\text{KL}(\pi^A || \pi^B) = \int_{\pspace} \pi^A(\param) \log\left( \frac{\pi^A(\param)}{\pi^B(\param)}\right)\ d \mu_\pspace.
\end{equation}
When using the population-informed prior for Bayesian inference, we can express the information gained from the initial distribution to the population-informed posterior as,
\begin{eqnarray}\nonumber
    \text{KL}(\postupdens || \priordens) &=& \int_{\pspace} \frac{1}{\tilde{C}}r(\param)\likedens({\bm y} | \param) \priordens(\param) \log\left( \frac{\frac{1}{\tilde{C}}r(\param)\likedens({\bm y} | \param)\priordens(\param)}{\priordens(\param)}\right) \ d \mu_{\pspace},\\\label{eq:kl_pop_post_pri}
    &=& \int_{\pspace} \frac{1}{\tilde{C}}r(\param)\likedens({\bm y} | \param) \log\left( \frac{1}{\tilde{C}}r(\param)\likedens({\bm y} | \param)\right) \ d \priormeas.
\end{eqnarray}
By collecting certain terms above, we can also express~\eqref{eq:kl_pop_post_pri} in terms of the standard Bayesian posterior as
\begin{equation}\label{eq:KLpostup}
    \text{KL}(\postupdens || \priordens) = \int_{\pspace} \frac{C}{\tilde{C}}r(\param)\postdens(\param) \log\left( \frac{\frac{C}{\tilde{C}}r(\param)\postdens(\param)}{\priordens(\param)}\right) \ d \mu_{\pspace},
\end{equation} 
which is similar to the KL divergence from the prior to the standard posterior,
\begin{equation*}
    \text{KL}(\postdens || \priordens) = \int_{\pspace} \postdens(\param) \log\left( \frac{\postdens(\param)}{\priordens(\param)}\right) \ d \mu_{\pspace},
\end{equation*}
but we see that~\eqref{eq:KLpostup} contains an additional multiplier term, namely, $\frac{C}{\tilde{C}}r(\param)$.
It is worth emphasizing that using the KL divergence as a measure of information gain of the posterior from a prior or initial density does not ensure that for every realization of individual-level data the resulting population-informed priors produce an increase in information gain in comparison to standard Bayesian inference.
In fact, even in the linear-Gaussian case, it is possible that for 
some realizations of individual-level data, population-informed priors result in decreased information gain measured by the relative difference in KL divergences given as 
\begin{eqnarray}\label{eq:rel_inf}
    \frac{\text{KL}(\postupdens || \initdens)- \text{KL}(\postdens || \initdens)}{\text{KL}(\postdens || \initdens)}.
\end{eqnarray}
To understand this phenomenon, consider the KL divergence between two arbitrary Gaussian random variables, $p \sim \mathcal{N}\left({\bm m}_p, {\bm \Gamma}_p\right)$ and $q \sim \mathcal{N}\left({\bm m}_q, {\bm \Gamma}_q\right)$ given as
\begin{eqnarray}\label{eq:KL_analyt}
    \text{KL}(p || q) = \frac{1}{2} \left[ \log \frac{|{\bm \Gamma}_q|}{|{\bm \Gamma}_p|} - k + ({\bm m}_p - {\bm m}_q)^\top {\bm \Gamma}_q^{-1}({\bm m}_p - {\bm m}_q) + \text{Tr}\left({\bm \Gamma}_q^{-1}{\bm \Gamma_p}\right)\right],
\end{eqnarray}
where $k = \text{dim}({\bm m}_p) = \text{dim}({\bm m}_q)$.
From~\eqref{eq:KL_analyt}, we see dependence of the divergence on both the mean and covariance associated with the random variables. 
In the context of Bayesian inference, this equates to dependence upon the posterior mean, which itself depends upon the realization of data used to solve the inverse problem (see~\eqref{eq:linGaus_mean} and~\eqref{eq:linGaus_pop_mean}).
Thus, even if the trace and determinant of the population-informed covariance is reduced, for some realizations of data, one may see a larger distance between the posterior and prior means (measured by the inner product $\langle {\bm m}_p - {\bm m}_q, {\bm m}_p - {\bm m}_q \rangle_{\Gamma_{q}^{-1}}$), which corresponds to a larger KL divergence for standard inference in comparison to population-informed inference.
However, the numerical results given in~\Cref{sec:computational_results} indicate that for an overwhelming majority of realizations of data, information gain increases when utilizing population-informed priors.
Furthermore, the expected (with respect to realizations of data) gain in information is larger for population-informed inference than standard inference.

\section{Computational Results}
\label{sec:computational_results}

We present computational examples that provide intuition for the properties of population-informed inference. 
The first two examples are used to numerically support the theoretical results of the prior section for the linear-Gaussian case while the third example involves a nonlinear structural mechanics model for an additive manufacturing problem that illustrates the effectiveness of population-informed inference in a nonlinear, non-Gaussian setting.

This section is organized as follows:
\Cref{sec:linGauss_comp_same} provides results for linear parameter-to-observable maps with Gaussian prior and likelihood models in the scenarios where the individual and population models are the same; \Cref{sec:linGauss_comp_diff} considers differing individual and population models. 
We conclude in~\Cref{sec:struc_mech_comp} with an additive manufacturing exemplar motivated by digital twin applications, where parameter-to-observable maps are nonlinear and dependent upon finite-element structural mechanics models.   

\subsection{The Linear-Gaussian Case with Identical Population and Individual Models}
\label{sec:linGauss_comp_same}
Consider first a scenario where the population and individual models (i.e., $f_p$ and $f_i$) are the same. 
Such a scenario can occur, for instance, when population data are aggregated from prior measurements of individuals using the same type of individual measurements that will be collected in the future. 
To make this concrete, suppose the individual and population parameter-to-observable maps are both given by
\begin{eqnarray}\label{eq:comp_ind_mod}
    f_i(\bml) = f_p(\bml) = {\bm B}\bml = \begin{bmatrix}
        2 & -1 
    \end{bmatrix} \begin{bmatrix}
        \lambda_1 \\ \lambda_2
    \end{bmatrix}.
\end{eqnarray}
Assume the data used to solve the Bayesian inverse problem are obtained via
\begin{eqnarray}\label{eq:stat_linGauss}
    y = {\bm B}\bml + \eta, \quad \eta \sim \mathcal{N}\left(0, 0.1
\right).
\end{eqnarray}
Assume that the initial distribution on the uncertain parameters, which will also be utilized in a standard Bayesian inference, is given by
\begin{eqnarray}\label{eq:com_prior}
    \pi_{\Lambda}^{\text{init}}(\bml)  \sim \mathcal{N}\left(\begin{bmatrix}
        0.4 \\ 0 
    \end{bmatrix}, \begin{bmatrix}
        0.15 & 0.0 \\ 0.0 & 0.15
    \end{bmatrix}
\right).
\end{eqnarray}
We explore the construction of a population-informed prior by solving the data-consistent inverse problem.

Note that synthetic data is generated according to~\eqref{eq:stat_linGauss}. 
Since the aim is to investigate the relative information gain over realizations of likely individual data, we specify a \textit{true} population-generating distribution $\popgendens$
from which to sample ${\bm \lambda}$. 
We note that when simulating likely data to analyze Bayesian inverse problem results, one often samples the prior.
Since our aim is to compare the performance of population-informed priors versus standard prior specification, we specify a \textit{true} distribution simply to serve as a construct for generating synthetic data.
In practice, this true distribution would not be known nor required as one would be given data on the individual with which to solve the Bayesian inverse problem.
Let 
\begin{eqnarray}\label{eq:comp_pop_gen}
    \popgendens(\bml) \sim \mathcal{N}\left(\begin{bmatrix}
    0.2 \\ 0.3
\end{bmatrix}, \begin{bmatrix}
    0.06 & 0 \\ 0 & 0.06
\end{bmatrix}
\right),
\end{eqnarray}
which results in an 
an observed density of
$\obsdens(f_p(\bml)) \sim \mathcal{N}\left(0.1, 0.3\right)$.
With this setup, and using~\eqref{eq:cov_up} and~\eqref{eq:mu_up}, the Gaussian updated density (which serves as the population-informed prior) is given by
\begin{eqnarray}\label{eq:comp_pop_pri}
\updens(\bml) \sim \mathcal{N}\left(\begin{bmatrix}
    0.12 \\ 0.14
\end{bmatrix}, \begin{bmatrix}
    0.078 & 0.036 \\ 0.036 & 0.132
\end{bmatrix}
\right).
\end{eqnarray}

\begin{table}
\begin{center}
\begin{tabular}{lll}
\hline
    & Determinant & Trace  \\ \hline
Standard            & 377.8         &    63.3  \\ \hline
Population-informed & {\bf 444.4}   & {\bf 73.3}       \\
\hline
\end{tabular}
\caption{Comparison of the determinant and trace for the inverse covariance matrices corresponding to the standard Bayesian posterior versus the population-informed posterior for the linear-Gaussian case with identical population and individual models.}
\label{tab:lingauss}
\end{center}
\end{table}

Since the covariances of the standard and population-informed inferences (and thus, their traces and determinants) do not depend on the realization of the individual data, we can immediately compare these covariances as a means of analyzing the increase in information gained through the use of the updated distribution as a population-informed prior.
To that end, we use~\eqref{eq:linGaus_cov} and~\eqref{eq:linGaus_pop_cov} to compute the covariances of the standard and population-informed posteriors and summarize in Table~\ref{tab:lingauss} the determinants and traces of the inverse covariances.
As guaranteed in~\Cref{sec:proof_linGaus}, the population-informed inverse covariance has a larger determinant and trace in comparison to the standard Bayesian inverse covariance.
In this case, the differences range roughly from 16-18\%;
note that large differences are not expected since the population and individual models being identical implies that each model is informing the same directions in the parameter space.
Thus, using the updated density as the population-informed prior simply gives the inference process a ``head-start''.

We now explore how population-informed inference compares to standard inference across multiple realizations of individual-level data.
We again utilize~\eqref{eq:com_prior} and generate $1\text{e}5$ realizations of individual data using~\eqref{eq:stat_linGauss} and the true population-generating distribution $\popgendens$.
We then solve the standard Bayesian and population-informed inference problems for each realization of data.
\Cref{fig:KL_lin_gauss} provides a histogram of the relative gain in information, measured by the KL divergence, from utilizing population-informed priors.
These results are computed utilizing~\eqref{eq:rel_inf}.
\begin{figure}[h]
    \centering
    \includegraphics[width=0.74\textwidth]{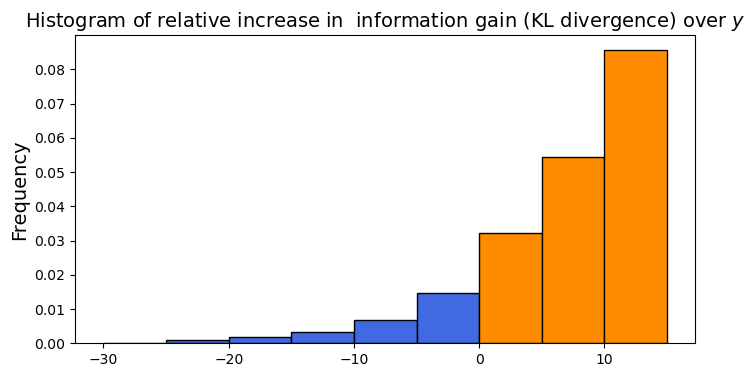}
    \caption{Histogram of the relative increase in KL divergence from the posterior to the prior (computed according to~\eqref{eq:rel_inf})
    for population-informed inference across $100,000$ realizations of data.}
    \label{fig:KL_lin_gauss}
\end{figure}
We note that in~\Cref{fig:KL_lin_gauss}, there are some realizations of data ($ < 7\%$) for which population-informed inference results in a decrease in the information gain in comparison to standard Bayesian inference.
However, for the vast majority of realizations of data ($>93\%$), information gain increases.
Furthermore, the distribution of information gain clearly skews towards larger values, which indicates that larger information gains occur at higher frequencies in this example. 

To better understand why some realizations of individual-level data lead to a reduction in the information gain, consider~\Cref{fig:out_of_dist}, which depicts the predicted (i.e., the push-forward of the initial through $f_i$) and the observed densities.
The green dots in~\Cref{fig:out_of_dist} indicate the individual data that led to a  decrease in information gain when utilizing the population-informed prior.
\begin{figure}[h]
    \centering
    \includegraphics[width=0.7\textwidth]{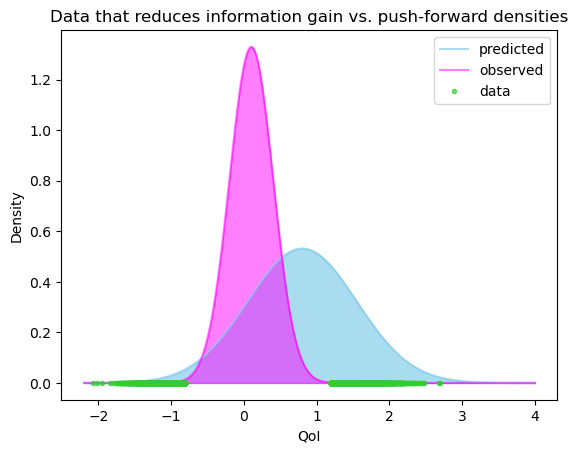}
    \caption{Push-forward of both the standard prior~\eqref{eq:com_prior} and the population-informed prior~\eqref{eq:comp_pop_pri} through the QoI defined by the individual model given in~\eqref{eq:stat_linGauss} versus the individual data that result in decreased information gain (computed according to~\eqref{eq:rel_inf}) when utilizing population-informed priors.} 
    \label{fig:out_of_dist}
\end{figure}
We see that data that are at the tail-ends of the observed distribution (i.e., the individual data we are least likely to observe), correspond to the cases where the population-informed prior negatively biases the resulting Bayesian inference.
Since the individual and population maps coincide (i.e., $f_p = f_i$), it is straightforward to assess the likelihood of such negative bias prior to computing the population-informed prior by comparing the individual data to the observed density.
This is analogous to determining if data are well supported by looking at a prior predictive in a standard Bayesian inverse problem.
Should the maps not coincide (i.e., $f_p \neq f_i$), one could 
predict negative bias after estimating the population-informed prior by comparing the individual data to the push-forward of the updated density.
This represents a trade-off where such a test would prevent subsequent Bayesian inferences from utilizing an improper prior at the computational cost of first solving the DCI problem.

\subsection{The Linear-Gaussian Case with Differing Population and Individual Models}
\label{sec:linGauss_comp_diff}
We now consider a scenario where the population and individual models differ, i.e., $f_p\neq f_i$.
Such scenarios can occur, for example, when  different experiments are performed on the population compared to the individuals leading to different data collection efforts.
To make this concrete, suppose the population parameter-to-observable map is given as
\begin{eqnarray}\label{eq:comp_pop_mod}
    f_p(\bml) = {\bm A}\bml = \begin{bmatrix}
        1 & 3
    \end{bmatrix} \begin{bmatrix}
        \lambda_1 \\ \lambda_2
    \end{bmatrix}.
\end{eqnarray}
For the individual parameter-to-observable map and corresponding statistical model on the data, we use~\eqref{eq:comp_ind_mod} and~\eqref{eq:stat_linGauss}, respectively. 
With the initial distribution, population-generating distribution, and corresponding observed densities given in~\Cref{sec:linGauss_comp_same} (see~\eqref{eq:com_prior} and~\eqref{eq:comp_pop_gen}), respectively, the resulting updated density is given as 
\begin{eqnarray}\label{eq:comp_pop_pri_diff}
\updens(\bml) \sim \mathcal{N}\left(\begin{bmatrix}
    0.37 \\ -0.09
\end{bmatrix}, \begin{bmatrix}
    0.138 & -0.036 \\ -0.036 & 0.042
\end{bmatrix}
\right).
\end{eqnarray}


With the goal of comparing the standard Bayesian posterior with the population-informed posterior computing using~\eqref{eq:comp_pop_pri_diff}, we consider one realization of individual data simulated according to~\eqref{eq:stat_linGauss}.
In the linear-Gaussian setting, the posterior covariance does not depend upon the realization of data, but for visualizing the contours of the posterior, we set $y \approx 0.39$ when solving the Bayesian inference problem.
\Cref{fig:stand_v_pop_post_linGauss} plots the
posterior contours corresponding to standard inference (left) versus population-informed inference (right), where the covariance matrices are
\begin{eqnarray}\label{eq:com_post_covs}
    \tilde{\bm \Gamma}_{\text{post}} \approx
    \begin{bmatrix}
         0.0218 & 0.00644 \\  0.00644 & 0.0265 
    \end{bmatrix}
    \quad \text{and} \quad
    {\bm \Gamma}_{\text{post}} \approx
    \begin{bmatrix}
         0.044 & 0.0529 \\  0.0529 & 0.124
    \end{bmatrix}.
\end{eqnarray} 
%
As supported by~\Cref{tab:lingauss2}, utilizing the population-informed prior results in significant increases to both the trace and determinant of the inverse covariance matrix. 
This was largely due to population and individual models informing different directions in the parameter space.
Moreover, both \Cref{fig:stand_v_pop_post_linGauss} and the covariances given in~\eqref{eq:com_post_covs} indicate that leveraging population data provides a greater reduction in posterior uncertainty, especially with respect to $\lambda_2$. 
This reduction results from the fact that in the Bayesian inverse problem, the parameters $[\lambda_1, \lambda_2]$ are not mutually identifiable; from~\eqref{eq:comp_ind_mod}, one can understand that each value of $f_i$ corresponds to a contour of values in the parameter space.
Thus, leveraging population data to construct a population-informed prior constrains dimensions of the parameter space complementary to this contour, leading to an overall reduction in parameter uncertainty when solving the Bayesian inverse problem. 

\begin{table}
\begin{center}
\begin{tabular}{lll}
\hline
    & Determinant & Trace  \\ \hline
Standard            & 377.8         &    63.3  \\ \hline
Population-informed & {\bf 1862.2}   & {\bf 90.0}       \\
\hline
\end{tabular}
\caption{Comparison of the determinant and trace for the inverse covariance matrices corresponding to the standard Bayesian posterior versus the population-informed posterior for the linear-Gaussian case with differing population and individual models.}
\label{tab:lingauss2}
\end{center}
\end{table}

\begin{figure}[h]
    \centering
    \includegraphics[width=0.48\textwidth]{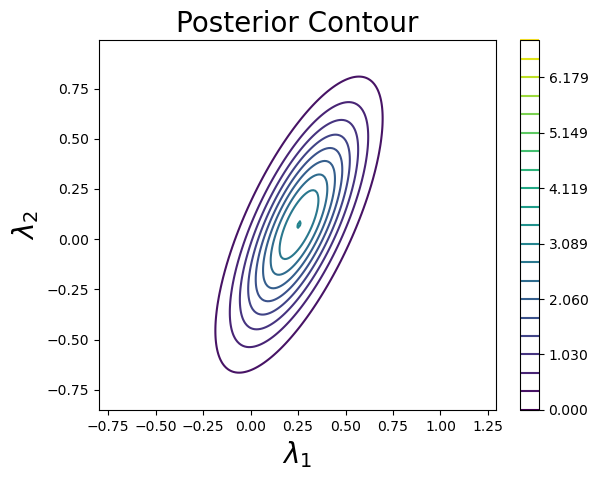}
    \includegraphics[width=0.48\textwidth]{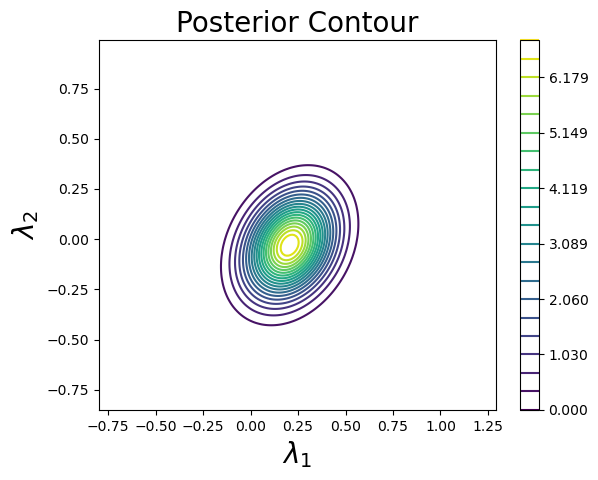}
    \caption{Contour plots of the Bayesian posterior corresponding to standard inference (left) versus population-informed inference (right) for a linear-Gaussian model with differing population and individual models and synthetically generated datum $y = 0.39$.} 
    \label{fig:stand_v_pop_post_linGauss}
\end{figure}

When looking at the relative increase in information gain resulting from population-informed inference, we find that across $100,000$ realizations of individual data (generated according to~\eqref{eq:stat_linGauss} where ${\bm \lambda} \sim \popgendens$), information gain always increases when utilizing population-informed priors. 
\Cref{fig:KL_lin_gauss_diff} depicts the histogram of relative increases in information gain across the realizations of data, where the expected (with respect to $y$) percent increase in information gain is approximately $53\%$.
\begin{figure}[h]
    \centering
    \includegraphics[width=0.74\textwidth]{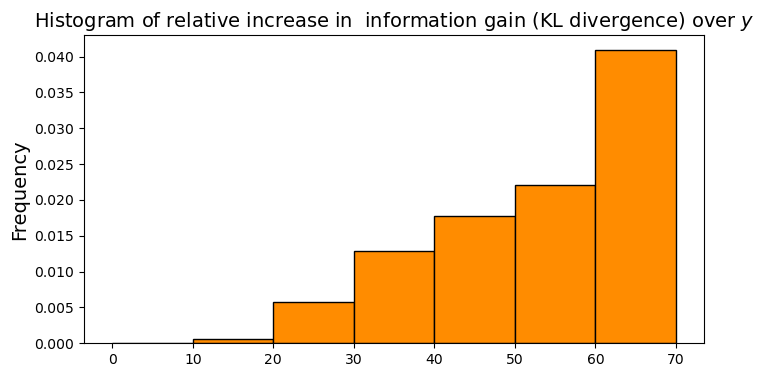}
    \caption{Histogram of the relative increase in KL divergence from the posterior to the prior (computed according to~\eqref{eq:rel_inf})
    for population-informed inference across $100,000$ realizations of data using a linear-Gaussian model with different population and individual models.}
    \label{fig:KL_lin_gauss_diff}
\end{figure}

Note that in many applications of inference, the Bayesian inverse problem is ill-posed.
Thus, even as the amount of individual data collected goes to infinity, the Bayesian posterior collapses only to a manifold for which the (individual) data is informative.
However, the results depicted by~\Cref{fig:stand_v_pop_post_linGauss,fig:KL_lin_gauss_diff} highlight how leveraging data on a related population -- resulting from experiments or measurements that differ from the individual -- can better constrain the Bayesian inverse problem and ultimately further reduce parameter uncertainty.

\subsection{Computational Results for Nonlinear Structural Mechanics Models}
\label{sec:struc_mech_comp}
We now present a computational study of an additive manufacturing exemplar motivated by digital twins that demonstrate the effectiveness of constructing population-informed priors even in the nonlinear, non-Gaussian setting for a larger-scale, real-world problem.
A digital twin provides a virtual representation of a physical asset, which can be updated with data allowing one to monitor the health and evolution of the asset over its lifespan. 
Such digital representations are often comprised of mathematical models, meshes, and parameters specific to the physical asset to enable accurate representations of the individual asset. 
Furthermore, digital twins are often used to represent fleets of similar individuals/assets, such as a field of wind turbines, a cohort of patients receiving healthcare, etc. 
One challenge in such applications is that data on an individual may be limited; thus, updating the digital representation (potentially in real-time) through Bayesian inference can be challenging. 
We demonstrate how leveraging data from the related assets can improve inference on the individual
for an additive manufacturing scenario bearing a problem structure similar to those of digital twin applications. 

The additive manufacturing problem considered herein is focused on the manufacturing of steel ``dog-bone'' structures at 
Sandia's Laser Engineered Near Net Shaping (LENS\textsuperscript{\tiny\textregistered}) facility~\cite{HECKMAN2020138632}.
Although one may expect the additively manufactured structures to be identical, there exist variations in material properties across individual structures in a given population of structures due to the manufacturing process.
The right plot of~\Cref{fig:ind_pop} illustrates the impact this inherent variability has on the behavior of individual structures during automated high-throughput tensile testing.
\begin{figure}[h]
    \centering
    \includegraphics[scale=0.5]{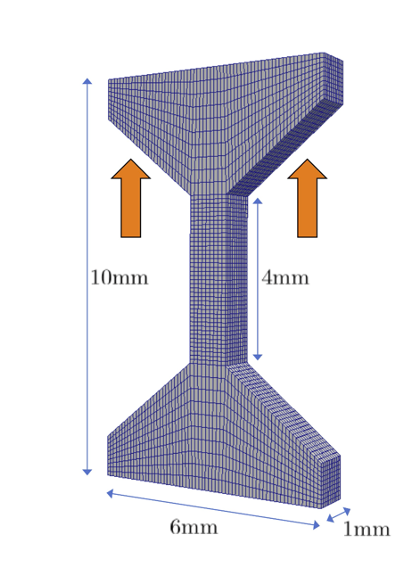}
    \includegraphics[width=0.48\textwidth]{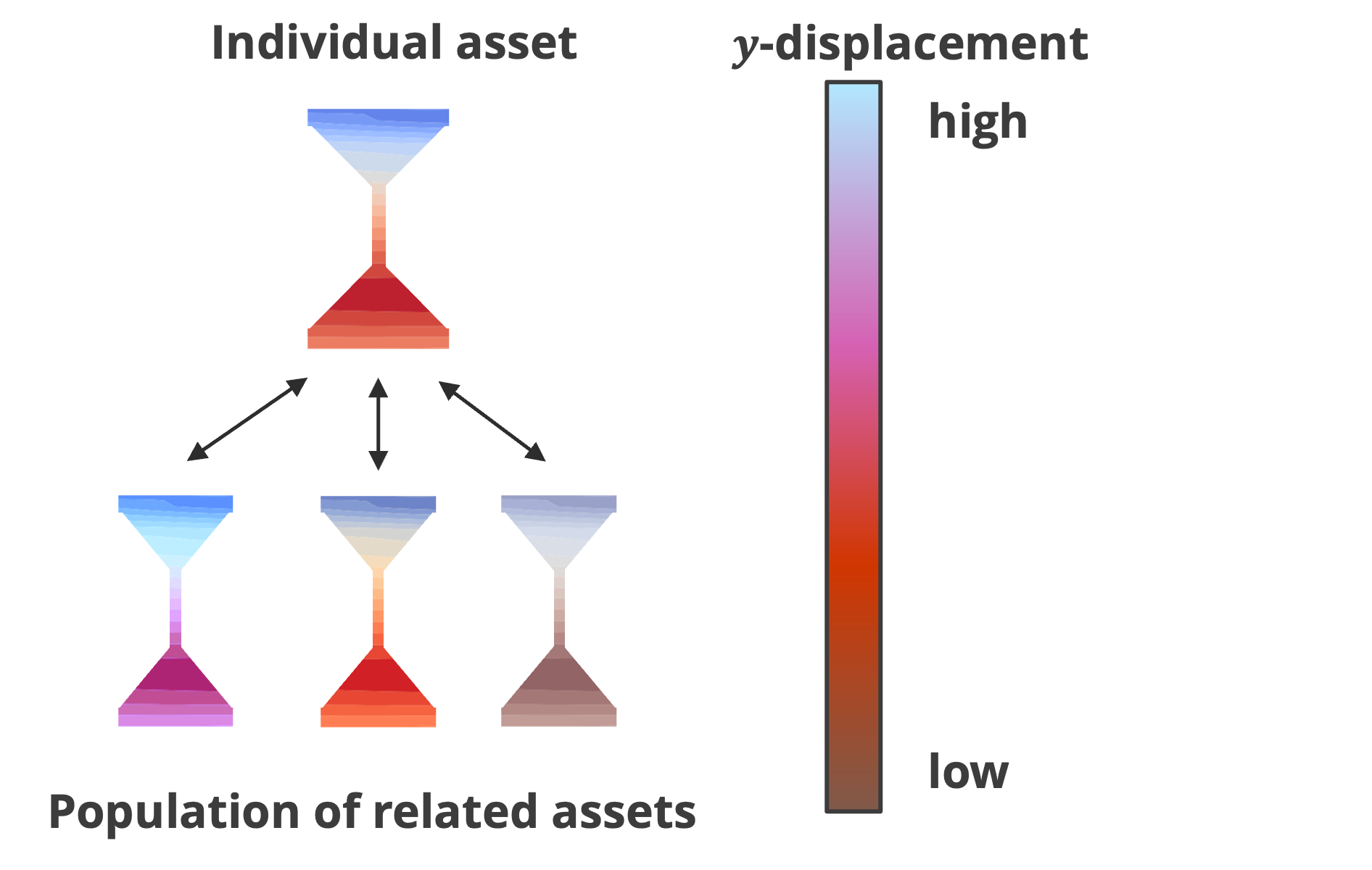}
    \caption{Digital representation of the dog-bone structure given by a finite element computational model (left) and depiction of population of dog-bone structures with varying average $y$-displacements under tensile loading (right).
    } 
    \label{fig:ind_pop}
\end{figure}
Here, we consider a digital representation of each dog-bone structure given by a finite element computational model of the structures, where we seek to infer the specific material properties determining the underlying linear elasticity models from observational data on the physical asset.
The left plot of~\Cref{fig:ind_pop} depicts the digital representation of a dog-bone structure given by a finite element computational model.

For this application, data collection on each individual are limited as the tests can be destructive, multiple tests cannot be performed in succession for an individual, nor can multiple quantities of interest (QoI) be measured simultaneously for an individual.
With this said, we investigate tensile testing that applies a load to the top portion of the dog-bone structure and records deformations in the $x$, $y$, and $z$ directions.  
For the sake of simplicity, we use synthetic data and model QoI in the linear elastic regime, i.e., we do not model fracture or failure of the structures.
Moreover, we use data on two QoIs to infer the Lam\'e parameters, Young's modulus $E$ and Poisson's ratio $\nu$. 
Specifically, the QoI used are the average $y$ displacement across the connection bar (thin bar connecting the top and lower segments of the dog-bone structure) as well as the sum of the average $x$ and $z$ displacements over the connection bar. 

The tensile loading of the dog-bone structures is modeled using the Multi-Resolution Hybridized Differential Equations (MrHyDE) software package developed at Sandia National Labs~\cite{MrHydeUser,MrHydeCopyright}.
Specifically, MrHyDE is used to generate data by solving the following linear elastic model on the 3-dimensional computational domain $\Omega$,
\[
\begin{cases}
-\nabla \cdot \bm{\sigma}({\bm u}) = \bm{f}, & x \in \Omega,\\
{\bm u} = \bm{g}_D, & x \in \Gamma_D \subset \partial \Omega, \\
\bm{\sigma}({\bm u})\cdot {\bm n} = \bm{g}_N, & x \in \Gamma_N \subset \partial \Omega, \\
\end{cases}
\]
where ${\bm u}$ represents displacements in the coordinate directions, ${\bm f}$ is the volumetric force, $\bm{g}_D$ is the fixed Dirichlet condition on a portion of the boundary, $\bm{g}_N$ is the force applied on the Neumann portion of the boundary, and
${\bm n}$ is the outward facing normal.
The constitutive relation is assumed to be linear and isotropic so that it can be expressed in terms of the Lam\'{e} parameters:
\[\bm{\sigma}({\bm u}) = \lambda (\nabla \cdot {\bm u}) \mathbb{I} + 2\mu \epsilon({\bm u}), \quad \epsilon({\bm u}) = \frac{1}{2}(\nabla {\bm u} + \nabla {\bm u}^T ).\]

Here, we model the individual data as \begin{eqnarray*}
    f_i({\bm \lambda}) = \frac{1}{C_{\text{bar}}} \int_{I_{\text{bar}}} \left( |u_x| + |u_z|\right) d \mathbf{x},
\end{eqnarray*}
where $u_x$ and $u_z$ are the $x$ and $z$ components of the displacement vector respectively, $I_{\text{bar}}$ is the bar region defined by $-0.05 \leq y \leq 0.05$ with $-0.0254 \leq x \leq 0.0254$ and $-0.0254 \leq z \leq 0.0254$, and $C_{\text{bar}} = 6.4516$E-5 is the volume of the bar. 
The data from an individual dog bone satisfies 
\begin{eqnarray}\label{eq:ind_data_db}
    y = f_i(\bml) +  \eta, \quad \eta \sim \mathcal{N}\left(0, \sigma^2_{\text{noise}}
\right),
\end{eqnarray}
where $\sigma_{\text{noise}} = 4.15\text{e-}07$.
Note that in problems where the noise variance cannot be determined experimentally, it is straightforward to estimate the variance alongside model parameters using both standard and population-informed priors.
Known properties of steel are used to construct the the prior/initial, yielding
\begin{eqnarray}\label{eq:dog_bone_pri}
    \pi_{\Lambda}^{\text{init}}(\bml) &=&= \pi_{\Lambda}^{\text{init}}(E) \times \pi_{\Lambda}^{\text{init}}(\nu)
    = \mathcal{U}_{\left[180, 210\right]} \times \mathcal{U}_{\left[0.25, 0.35\right]},
\end{eqnarray}
where we use the standard relationships to map $E$ and $\nu$ to $\lambda$ and $\mu$:
\[
  \lambda = \frac{E\nu}{(1+\nu)(1-2\nu)}, \quad
  \mu = \frac{E}{2(1+\nu)}.
\]
The displacements of the dog-bone structures are fixed to be zero on the bottom boundary and an upward force is applied on the sides as shown in Figure~\ref{fig:ind_pop}.
All other sides use stress-free boundary conditions, e.g., $\bm{g}_N=\bm{0}$, and no volumetric force is used, i.e., $\bm{f}=\bm{0}$.
Additionally, the finite element discretization based on tri-linear conforming finite elements with a mesh containing $63,362$ hexahedral elements is used.
The displacement fields and the von Mises stress are shown in Figure~\ref{fig:dog_bone_nom} for one realization of the random material properties.
\begin{figure}[h]
    \centering
    \includegraphics[width=0.85\textwidth]{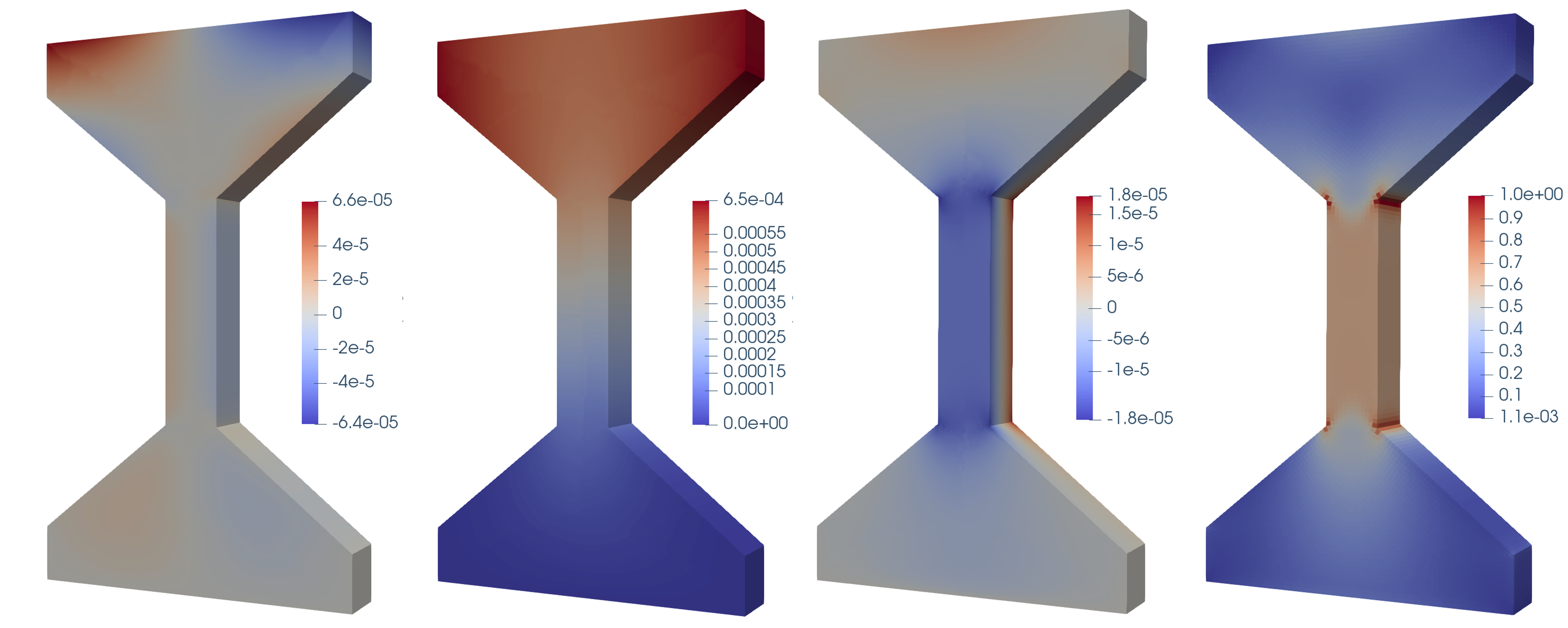}
    \caption{From left-to-right, the x-displacement, y-displacement, z-displacement, and von Mises stress for one realization of the dog-bone structures.
    } 
    \label{fig:dog_bone_nom}
\end{figure}

To solve the relevant inverse problems, we generate $4\text{e}04$  
realizations of the material properties and extract the corresponding QoI from the simulations.
Then, we perform the standard
Bayesian inference of an individual dog-bone structure using the prior defined in~\eqref{eq:dog_bone_pri} and experimental data simulated according to~\eqref{eq:ind_data_db}.
Here, a value of $y \approx 1.3\text{e-}05[\text{mm}]$ 
is used.
The middle plot of~\Cref{fig:dog_bone_tri_plot} depicts the resulting posterior distribution on the Lam\'e parameters. The individual data clearly informs Poisson's ratio, but the uncertainty regarding Young's modulus is not significantly
reduced.
\begin{figure}[h]
    \centering
    \includegraphics[trim={0 0 0 1cm},clip,width=0.98\textwidth]{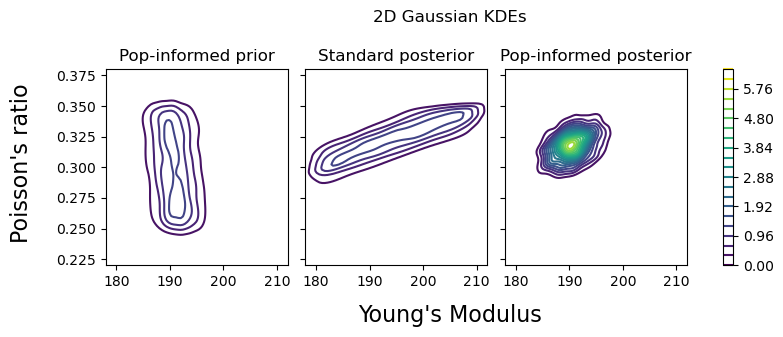}
    \caption{Contour plots of the population-informed prior (left), the standard Bayesian posterior (middle), and the population-informed posterior (right) all approximated using 
    kernel density estimation (KDE).
    } 
    \label{fig:dog_bone_tri_plot}
\end{figure}

Next we incorporate information from experiments performed on a population of related assets. Specifically, population data on the  average $y$ displacement over the connection bar is used to specify the
the observed density by $\pi_{\text{obs}}(f_p(\bml)) \sim \mathcal{N}\left(2.8\text{e-}04, 4.20\text{e-}06\right)$. 
Note that here we do not specify a true population-generating distribution as this would not be known in practice and is only necessary for synthetically generating a large number of realizations of individual data.
We then use DCI to estimate the population-informed prior depicted in the left plot of~\Cref{fig:dog_bone_tri_plot}.
Note that rejection sampling (see Algorithm~\ref{alg:alg_pop_post}) is used to generate samples from the updated density as closed-form expressions for the predicted and updated densities do not exist in this nonlinear, non-Gaussian setting.
The left plot of~\Cref{fig:dog_bone_tri_plot} demonstrates that the population data constrains  uncertainty in the Young's Modulus. Moreover, population-informed inference results in a significant increase in information gain, approximately $90\%$ as measured by~\eqref{eq:rel_inf}. This can be seen by comparing the middle and right plots of~\Cref{fig:dog_bone_tri_plot}.

\section{Conclusions}
\label{sec:conclusions}

In many applications of interest, data on an individual can be limited. For example, in additive manufacturing, experiments may be destructive, while in biomedical applications we may be limited in the diagnostic tests that can be performed on a single patient.
However, in such scenarios, there often exist populations of related individuals/assets for which similar experiments are conducted.
This work considered one approach for leveraging data from related assets to inform inference on the individual.
Specifically, we used data-consistent inversion (DCI) to estimate population-informed priors from data on a population of related assets. 
We then compared information gained in the Bayesian inverse
problem when utilizing population-informed priors versus standard prior specification, where information gain is measured in terms of the trace and determinant of the inverse posterior covariance as well as the KL divergence of the posterior from the prior. 

For linear-Gaussian scenarios, we proved population-informed priors will always lead to increased information gain as measured by the trace and determinant of the inverse posterior covariance; this result holds regardless of whether the individual and population parameter-to-observable maps differ.
Additionally, when measuring information gain in terms of the KL divergence, we demonstrated numerically that population-informed priors lead to increased information gain for a vast majority of realizations of data. 
Furthermore, the magnitude of information gain across realizations of data provided a positive expected gain in information.
Lastly, our numerical results based on an additive manufacturing exemplar demonstrated that our approach for constructing population-informed priors is able to inform dimensions of the parameter space not informed through standard Bayesian inference.

This work provides an enabling technology for digital twin systems by establishing rigorous mathematical theory for incorporating information from related assets into model predictions.
Consequently, this work could impact a variety of application spaces that possess populations of individuals/assets, including biomedical and engineering applications.
However, to broaden applicability of this approach, future work will address the computational challenges associated with DCI in high-dimensional output spaces and the development of optimal experimental design strategies for judiciously collecting population and individual data.

\section{Acknowledgements}\label{sec:acknowledgements}
T.~Butler's work is supported by the National Science Foundation under Grant No.~DMS-2208460.
T.~Butler's work is also supported by NSF IR/D program, while working at National Science Foundation. 
However, any opinions, findings, conclusions, or recommendations expressed in this material are those of the authors and do not necessarily reflect the views of the National Science Foundation.
T.~Wildey's work was supported by the U.S. Department of Energy, Office of Science,
Office of Advanced Scientific Computing Research, Early Career Research Program.
R.~White's and J.~Jakeman's work was supported by the Laboratory Directed Research Development (LDRD) program at Sandia National Laboratories.

This paper describes objective technical results and analysis. Any subjective views or opinions that might be expressed in the paper do not necessarily represent the views of the U.S. Department of Energy or the United States Government.
Sandia National Laboratories is a multimission laboratory managed and operated by National Technology and
Engineering Solutions of Sandia, LLC., a wholly owned subsidiary of Honeywell International, Inc.,
for the U.S. Department of Energy’s National Nuclear Security Administration under contract DE-NA-0003525.

\bibliographystyle{siam}
\bibliography{refs,jakeman-refs}

\end{document}